\documentclass[10pt, twocolumn]{IEEEtran}
\usepackage{pifont}
\usepackage{cite}
\usepackage[dvips]{epsfig}
\usepackage{graphicx}
\usepackage{calligra}
\usepackage[T1]{fontenc}
\usepackage{bbold}
\usepackage{mathrsfs}

\usepackage{subfigure}
\usepackage{color}
\usepackage{amsmath}
\usepackage{cases}

\usepackage{amssymb}
\usepackage[justification=centering]{caption}
\usepackage{framed}

\usepackage{subfigure}
\usepackage{color}
\usepackage{amsmath}
\usepackage{bm}
\usepackage{amssymb}
\usepackage{amsbsy}
\usepackage{bbm}
\usepackage{dsfont}
\usepackage[normalem]{ulem}
\usepackage{xcolor,cancel}
\usepackage{enumerate}
\usepackage{paralist}

\usepackage{algorithm, algorithmic}

\usepackage{amsthm}
\usepackage{fancyhdr}

\newtheorem{theorem}{Theorem}

\newtheorem{condition}{Condition}

\newtheorem{corollary}{Corollary}

\newtheorem{definition}{Definition}

\newtheorem{proposition}{Proposition}

\newtheorem{lemma}{Lemma}

\newtheorem{remark}{Remark}

\newtheorem{claim}{Claim}

\newtheorem{assumption}{Assumption}

\newcommand{\bbP}{\mathbbm{P}}

\newcommand{\cL}{\mathcal{L}}

\newcommand{\cA}{\mathcal{A}}

\newcommand{\cS}{\mathcal{S}}

\newcommand{\bs}{{\bf{s}}}

\newcommand{\bdelta}{{\boldsymbol{\delta}}}

\newcommand{\ignore}[1]{{}}

\newcommand{\cbk}{\color{black}}

\newcounter{parentalgorithm}

\ifCLASSINFOpdf
\else
\fi

\begin{document}
%
\title{Vulnerability of Finitely-long Blockchains in Securing Data}
%
%
%

\author{Yiming Jiang  and Jiangfan Zhang, ~\IEEEmembership{Member,~IEEE}
\thanks{Y. Jiang and J. Zhang are with the Department of Electrical and Computer Engineering, Missouri University of Science and Technology, Rolla MO 65409 USA (e-mail: yjk7z@mst.edu, jiangfanzhang@mst.edu)} 
}



\maketitle
%
\begin{abstract}
Recently, blockchain has been  applied in various fields to secure data exchanges and storage in decentralized systems.
In a blockchain application where the task of the application which makes use of the data stored in a blockchain has to be accomplished by a time instant,
the employed blockchain is essentially finitely-long. 
In this paper, we consider a general finitely-long blockchain model which is generalized from most existing works on finitely-long blockchain applications, and take the first step towards characterizing the vulnerability of finitely-long blockchains in securing data against double-spending attacks. For the first time, we develop a general closed-form expression for the probability of success in launching a double-spending attack on a finitely-long blockchain. This probability essentially characterizes the vulnerability of  finitely-long blockchains. Then, we prove that the probability of success in launching a double-spending attack on a finitely-long blockchain is no greater than that on an infinitely-long blockchain, which implies that finitely-long blockchains are less vulnerable to double-spending attacks than infinitely-long blockchains. Moreover, we show that unlike infinitely-long blockchains which can be surely paralyzed by a 51\% attack,   finitely-long blockchains are more resistant to 51\% attacks.
\end{abstract}


\IEEEpeerreviewmaketitle

\section{Introduction}

As cryptocurrencies are increasingly gaining popularity in both the financial sector and  society at large, their underlying technology, referred to as blockchain, shows great potential for many applications in different engineering disciplines. 
Blockchain was firstly developed for financial applications \cite{nakamoto2008bitcoin, wood2014ethereum}, and has provided feasible measures for establishing mutual trust among network nodes and defending against security threats to data storage and exchanges.  
Because of its security-by-design and distributed nature without authority, blockchain has been applied to diverse engineering fields, such as smart grids \cite{kurt2019secure}, vehicular networks \cite{yang2018blockchain}, and smart city \cite{sharma2018distarch}.

As an emerging secure distributed database technology that revolutionizes the way information is secured, distributed, and shared, blockchain technology can eliminate the need for the central authority and operates on a peer-to-peer network with the following vital components \cite{puthal2018blockchain, dinh2018untangling}: (i) a chronologically ordered sequence of blocks that are cryptographically linked to each other and are shared, stored and synchronized over the network; (ii) strong cryptography enabling secure data storage and secure data exchanges, and (iii) a  consensus protocol that enables verification and validation of the authenticity and integrity of  stored and exchanged data, and thus enables mutual trust over the network instead of relying on a central authority. 
The cryptographic algorithms and digital signature algorithms of the blockchain technology can effectively prevent the impersonation of   network nodes and the attacks on the information exchanges among network nodes.

The consensus is a procedure through which  every new block that is added to the blockchain is the one and only version of the truth that is agreed upon by all network nodes, and hence, ensures that the local blockchain copies of the network nodes can reach a common agreement. In this way, consensus protocols can establish mutual trust among network nodes in a distributed computing environment without needing a central authority. 
There are several types of consensus protocols which can help blockchain network nodes achieve a common agreement, such as Proof of Work (PoW), Proof of Stake (PoS), Proof of Activity (PoA), and Proof of Capacity (PoC). The PoW requires every participant node to compete with each other to solve a computationally challenging puzzle to receive the right to add a new block to the blockchain.
The PoS involves the allocation of responsibility in validating and adding new blocks to the blockchain to a participant node in proportion to the number of cryptocurrency tokens held by it \cite{saleh2021blockchain}.  The PoC allows sharing of memory space of the contributing nodes in the network  \cite{debus2017consensus}. The more memory or hard disk space a node has, the more rights it is granted for maintaining the blockchain. The PoA is a hybrid that makes use of aspects of both PoW and PoS \cite{bentov2014proof}. 
Among all blockchain consensus protocols, the PoW is the most  widely used consensus protocol. For example, Bitcoin, Litecoin, ZCash, and Bitcoin Cash all adopt PoW as their consensus protocols \cite{ouyang2021pow,hopwood2016zcash,kwon2019bitcoin}. In this paper, unless otherwise noted, we assume that the considered blockchains employ the PoW consensus protocol. 

It is worth mentioning that the blockchain technology cannot eradicate all security threats to the data stored in a blockchain. As mentioned in \cite{nakamoto2008bitcoin}, blockchains are prone to the double-spending attack (DSA), which can possibly falsify the data stored in a blockchain and is considered one of the most devastating attacks against blockchains. In recent years, DSAs have occurred several times in financial blockchain applications and are expected to occur more often in the future. For example, one of Bitcoin forks, Bitcoin Gold, suffered double-spending attacks in 2018, and again in 2020, with more than 17 million U.S. dollars lost in total. For a blockchain application where the employed blockchain is utilized to  store data, a successful DSA can modify the existing data stored in the blockchain without being perceived, and hence seriously compromise any task which makes use of the data stored in the blockchain.

In a blockchain application, the employed blockchain can be infinitely-long or finitely-long. 
If the task of the application which makes use of the data stored in the employed blockchain has to be accomplished by some time instant $t_0$, then any future blocks and their data generated after $t_0$ do not affect the accomplished task. To this end, from the perspective of the task, the blockchain can be deemed to reach its final length and stop growing when the task has been accomplished, even though after the task has been accomplished, the blockchain in fact keeps growing to record future data which may be used for other tasks. 
In light of this,
we adopt the following definition to distinguish finitely-long blockchains from infinitely-long blockchains from the perspective of the task of the blockchain application.
\begin{definition}[Finitely-long Blockchain]  \label{Def_finitely_long_blockchain}
	For a blockchain application, if the task of the application which makes use of the past and present data stored in the employed blockchain has to be accomplished by some time instant, then from the perspective of this task, the blockchain employed in this application is deemed finitely-long.
\end{definition}

In  many financial blockchain applications, the employed blockchain plays the role of a ledger recording financial transactions, see \cite{nakamoto2008bitcoin} for instance. In view of  unceasing occurrence of transactions, the task of recording transactions has to be carried out over time and cannot be accomplished by a given time instant. Thus, the blockchain for this kind of application  should be deemed infinitely-long, which is widely adopted in many recent works, see \cite{nakamoto2008bitcoin,wood2014ethereum} for instance.
In contrast, in the blockchain applications for many engineering problems where their tasks have to be accomplished in a timely manner, the employed blockchains should be considered  finitely-long according to Definition \ref{Def_finitely_long_blockchain}. 
For example,  consider a blockchain-aided smart grid monitor system where the state of the power grid needs to be estimated based on meter measurements stored in a blockchain \cite{asefi2021application}. The estimation accuracy is generally determined by the number of meter measurements. Once the number of meter measurements stored in the blockchain is large enough to ensure that the estimation accuracy can meet a prescribed requirement, an estimate of the state of the power grid is produced, and hence, from then on, the growth of the blockchain can be regarded as terminated from the perspective of the task of state estimation since any future data stored in the blockchain will not affect the value of the produced state estimate. Another representative example of a finitely-long blockchain application is a blockchain-based e-voting system where once all ballots for the election have been securely recorded in the blockchain, the election result can be determined, and hence, the growth of the blockchain can be regarded as terminated from the perspective of the task of the election\cite{yang2021priscore}.  In some other existing works on blockchain applications, such as power auction \cite{hassan2021optimizing}, e-voting \cite{hjalmarsson2018blockchain}, and vehicular announcement network \cite{hassija2020traffic,li2018creditcoin}, the employed blockchains are essentially finitely-long as well according to Definition \ref{Def_finitely_long_blockchain}. 
It is worth mentioning that for a finitely-long blockchain application, if an attacker aims at undermining the task of the application by falsifying the data stored in the blockchain, then any effective attack has to be successfully launched before the blockchain reaches its final length. Thus, from the perspective of attackers, the timeliness of attacks which are launched on finitely-long blockchains matters. This is very different from infinitely-long blockchains which gives rise to significant differences between the vulnerabilities of finitely-long blockchains  and infinitely-long blockchains.

Lately, there has been a growing tendency to negligently assume that the data stored in a blockchain is perfectly secure in a lot of recent literature on finitely-long blockchain applications, see \cite {ling2019blockchain, lei2017blockchain, cha2018blockchain, lombardi2018blockchain} for instance.
However, similar to infinitely-long blockchains, the data stored in a finitely-long blockchain are also subject to DSAs, and therefore, may not be perfectly secured in adversarial environments. The basic idea of the DSA is to build a counterfeit branch storing falsified data in a blockchain, and extend this counterfeit branch to be longer than the authentic branch of the blockchain. Once the counterfeit branch becomes longer than the authentic branch of the blockchain, the DSA is deemed to be launched successfully since the counterfeit branch and the data stored in it are considered valid while the authentic branch of the blockchain and its data are considered invalid according to the longest chain protocol of blockchain \cite{nakamoto2008bitcoin}.

The DSA on infinitely-long blockchains has been studied in previous literature,  see \cite{ozisik2017explanation, rosenfeld2014analysis,grunspan2018double,zaghloul2020bitcoin} for instance. However, the analysis of DSAs on infinitely-long blockchains cannot be applied to finitely-long blockchains since there is a time limit for adversaries to attack finitely-long blockchains.
In particular, unlike infinitely-long blockchains, in order to falsify the data stored in a finitely-long blockchain, a DSA has to be launched successfully before the time instant that the task is accomplished, that is, the time instant that the authentic branch of the blockchain grows to a predetermined final length. To the best of our knowledge, there is no existing work studying the vulnerability of finitely-long blockchains and analyzing DSAs on  finitely-long blockchains. In this paper, we consider a general finitely-long blockchain model which is generalized from existing works on finitely-long blockchain applications. For the first time, we theoretically characterize the vulnerability of the finitely-long blockchain by developing the closed-form expression for the probability of success in launching a double-spending attack on a finitely-long blockchain.

\subsection{Summary of Results and Main Contributions}


Considering a finitely-long blockchain which is under a DSA, 
the vulnerability of the finitely-long blockchain can be characterized by the probabilities of success in launching a DSA on the finitely-long blockchain starting from different blocks of the authentic branch.
We show that the derivation of these probabilities can be cast as a two-sided boundary hitting problem for a two-dimensional random walk with two possible walking directions. In particular, the two boundaries in this problem are orthogonal to each other, while the two possible walking directions of the random walk are not orthogonal to each other. Moreover, one walking direction of the random walk is neither orthogonal nor parallel to any of the two boundaries. For such a two-sided boundary hitting problem,  the general closed-form expression for the probability that the random walk hits one boundary before the other is developed, which can describe the probability of success in launching a DSA on a finitely-long blockchain starting from any block of the authentic branch.
This probability depends on the normalized hash rate of the attacker who launches the DSA, which is defined as the ratio of the computational power owned by the attacker to the total  computational power of the blockchain network, the length of the authentic branch of the blockchain at the time of the attack, the ordinal number of the block where the DSA starts from, and the predetermined final length of the finitely-long blockchain.

Moreover, we theoretically compare the vulnerability of finitely-long blockchains with that of infinitely-long blockchains. To be specific, we prove that the probability of success in launching a DSA on a finitely-long blockchain is no greater than that on an infinitely-long blockchain, which implies that finitely-long blockchains are less vulnerable to DSAs than infinitely-long blockchains. It has been proved in previous literature that if the normalized hash rate of an attacker is greater than $50\%$ and the attacker  attacks an infinitely-long blockchain by launching a DSA, which is called a $51\%$ attack, then the probability of success in launching the DSA  is always one, which implies that the attacker is able to falsify any data stored in the infinitely-long blockchain, and hence the security of the infinitely-long blockchain is completely demolished \cite{nakamoto2008bitcoin}. In contrast, we show that unlike infinitely-long blockchains, even though the normalized hash rate of an attacker is greater than $50\%$, the probability of success in launching a DSA on any finitely-long blockchain is strictly less than one. This indicates that 51\% attacks cannot completely demolish the security of finitely-long blockchains. 

\subsection{Related Work}
Finitely-long blockchains have been recently integrated into diverse engineering applications to enhance data security, such as smart homes, smart grids, smart cities, and vehicle networks, see \cite{yang2021priscore,hjalmarsson2018blockchain,hassija2020traffic,lei2017blockchain, cha2018blockchain,lombardi2018blockchain,ling2019blockchain,li2018creditcoin,zhou2018beekeeper,yang2021privacy}  and the references therein. For example, in \cite{li2018creditcoin}, the authors propose a blockchain-based vehicular announcement network where the blockchain is employed to protect data against tampering and make them widely available and accessible over the network against the possibility of node failures and hacking. 
In this application, the employed blockchain is actually finitely-long because once the task of querying recent accident reports which are stored in a blockchain is accomplished, the employed blockchain can be deemed to reach its final length and any future accident report will not affect the accomplished querying task. 
A traffic jam probability prediction system is proposed in \cite{hassija2020traffic} where a blockchain helps vehicles to securely share and request for the live traffic  at a particular location. The blockchain employed in \cite{hassija2020traffic} is also finitely-long because once the task of prediction is accomplished, the growth of the blockchain can be considered to terminate from the perspective of the task of prediction  and future traffic information does not affect the accomplished predicted probability.  However, all the mentioned works assume the data stored in  finitely-long blockchains are perfectly secure, which is not true in general. In particular, the data stored in a finitely-long blockchain can be modified by DSAs, which motivates us to investigate DSAs on finitely-long blockchains.  

The DSA on infinitely-long blockchains has been well studied in previous literature, see \cite{ozisik2017explanation, rosenfeld2014analysis,grunspan2018double,zaghloul2020bitcoin} and the references therein. In \cite{nakamoto2008bitcoin}, the author points out infinitely-long blockchains are vulnerable to DSAs, and the author derives the probability of success in launching a DSA on an infinitely-long blockchain. Moreover, the author indicates that if an attacker controls more than half of the computational power of a blockchain network, the attacker can always successfully launch a DSA on the infinitely-long blockchain. Based on the results in \cite{nakamoto2008bitcoin}, the authors of \cite{grunspan2018double} provide a more detailed analysis on the properties of  DSAs on infinitely-long blockchains. However, all the previous literature which theoretically investigates the DSA  only focuses on infinitely-long blockchains, and  cannot apply to  finitely-long blockchains. In this paper, we consider  DSAs launched on  finitely-long blockchains, and develop the closed-form expression for the probability that an attacker successfully launches a DSA on a finitely-long blockchain.

The paper is organized as follows. In section \ref{Section_systemmodel}, the general finitely-long blockchain model and double-spending attack model are introduced. Section \ref{Section_Probability_DSA} analyzes the probability that an attacker successfully launches a DSA on a finitely-long blockchain. Numerical simulations are provided in Section \ref{Section_simulation}, and Section \ref{Section_conclusion} provides our conclusions.

\section{Finitely-long Blockchain and Adversary Models}\label{Section_systemmodel}
In this section, we first introduce a general finitely-long blockchain model which subsumes most models adopted by previous works on finitely-long blockchain applications, and then the double-spending attack model is elaborated.

\subsection{Finitely-long Blockchain Model} 
\label{Section_FLBM}
We consider a general finitely-long blockchain model which is generalized from most existing works on finitely-long blockchain applications, see \cite{yang2021priscore,hjalmarsson2018blockchain,hassija2020traffic,lei2017blockchain, cha2018blockchain,lombardi2018blockchain,ling2019blockchain,li2018creditcoin,zhou2018beekeeper,yang2021privacy} for instance.  
The working mechanism of  finitely-long blockchains is similar to that of infinitely-long blockchains proposed in \cite{nakamoto2008bitcoin}, except that  a finitely-long blockchain can be considered to stop growing when its longest branch reaches a predetermined final length.

Different nodes can play different roles in a finitely-long blockchain network.
The main duties of nodes include data generation, block/data routing, block/data verification, and block mining.  
In a finitely-long blockchain network, there are mainly two types of nodes, that is, wallets\footnote{The term wallet was firstly introduced in the Bitcoin network. That we use the same terminology here is because the duty of these nodes  is similar to that of wallets in the Bitcoin network which produce the data stored in the blockchain.} and miners \cite{zaghloul2020bitcoin}. 
The wallets produce data which can be transactions, sensor measurements, or any other information, while the main duty of miners is to generate blocks in the blockchain to store the data produced by the wallets. Note that it is possible that a node in a finitely-long blockchain network plays the roles of both wallet and miner. The working mechanism of a finitely-long blockchain network mainly includes two components which are presented below.
 \begin{figure}
	\centering   
	\includegraphics[width=2.5in]{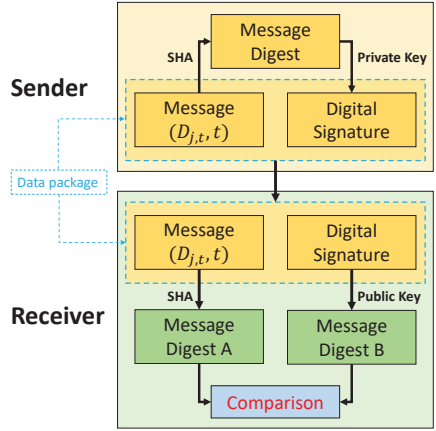}
	\caption{Inter-node data exchanges based on the asymmetric encryption mechanism.}
	\label{fig_encryption}
\end{figure}
\subsubsection{Data Exchanges} In a finitely-long blockchain network, wallets first transmit their data to every miner, and then the miners compete with each other to generate blocks to store wallets' data in the blockchain. The process of generating new blocks for the blockchain will be elaborated in Section \ref{Section_Block_Mining}. Once a miner generates a block, it sends the generated block to all the other nodes which retain local copies of the blockchain. In all these inter-node data exchanges, asymmetric encryption based on a public key infrastructure is used, which is illustrated in Fig. \ref{fig_encryption}. To be specific, each node in the blockchain owns a public-private key pair that forms the digital identity of the node. The public key is available at the other nodes, while the private key is only available at its owner. A secure hash algorithm (SHA), e.g., SHA-256 and SHA-512 \cite{pilkington2016blockchain}, is used in the data encryption process of every data exchange. To delineate the inter-node data exchanges, let's consider the data exchanges between a wallet and a miner as an example. Let $D_{j,t}$ denote the $t$-th data of the $j$-th wallet where time index $t$ where time index $t=1,2,...$ Before transmitting $D_{j,t}$ to miners,  the $j$-th wallet first processes the message which contains data $D_{j,t}$ and its time index $t$ by employing an SHA and obtains a message digest.  It then encrypts the message digest via its private key by using a digital signature algorithm, e.g., Elliptic Curve Digital Signature Algorithm \cite{dinh2018untangling}, and produces a digital signature. Finally, the $j$-th wallet transmits the data package consisting of the message and the corresponding digital signature to all the miners of the network.  

Once a miner receives a data package, it first decrypts the received digital signature via the public key of the wallet and obtains a message digest, which is signified by Message Digest B in Fig. \ref{fig_encryption}. Then, the miner processes the received message via the SHA and obtains another message digest, which is represented by Message Digest A in Fig. \ref{fig_encryption}. Only if these two message digests exactly match with each other, the authenticity of the received data package is verified and the received message will be used in the future processes. Otherwise, the received data package will be discarded and retransmission can take place. Note that the digital signature received at the miner can only be decrypted via the public key of the sender. Hence the miner can verify the identity of the sender to prevent  impersonation of the sender. Moreover, the SHA and the digital signature algorithm can  secure the data package transmission  since it is computationally intractable for an attacker to either find a different message which yields the same message digest or generate a valid digital signature for a fake message digest without the private key of the sender \cite{pilkington2016blockchain, dinh2018untangling}. To this end, the authenticity of the data packages received at  miners and the identities of  senders can be validated and secured.

\subsubsection{Block Mining and Consensus Protocol}\label{Section_Block_Mining} For each time $t$, every miner first collects data packages from all the wallets and verifies the authenticity of the collected data packages. Then every miner puts a header and all the wallets' data packages, that is,  the messages $(D_{j,t},t)$ for all $j$ and the digital signatures associated with these messages, into a new block.  The header consists of a discrete timestamp, the hash value of the last block (parent block) of the longest valid branch in the miner's local copy of the blockchain, the Merkle Root which is the root of the Merkle tree constructed by recursively hashing pairs of data packages until there is only one hash \cite{nakamoto2008bitcoin}, and a number called nonce which is the solution to a PoW puzzle.  The hash value of the parent block which is stored in the header of  the new block in essence cryptographically links the new block to its parent block. A simplified structure of a block is illustrated in Fig. \ref{fig_BlockStructure}. 

Next, the miners compete with each other in solving a difficult PoW puzzle for their new blocks, which is called mining. In Particular, each miner attempts to solve a PoW puzzle by searching for a valid nonce for its new block, via brute-force search, which renders the hash value of the new block with no less than a prescribed number of prefix zeros \cite{nakamoto2008bitcoin}. The difficulty of the PoW puzzle increases as the prescribed number of prefix zeros increases. When a miner solves its PoW puzzle first among all miners, i.e., finds a valid nonce value which meets the requirement, we say that the miner successfully mines its new block, and the miner broadcasts its new block to all the other nodes which retain local copies of blockchain. After receiving the newly mined block, the other nodes carry out a block validation procedure. Specifically, the other nodes first verify the authenticity of the messages in the received block, and confirm that the time index contained in each message of the received block is just one greater than that in the parent block of the received block. Then they verify that the hash value of the received block indeed has no less than  the prescribed number of prefix zeros. If the received block can pass this block validation, all the other nodes will add this received block after its parent block in their local copies of blockchain, and switch to work on solving the PoW puzzle for the next block. The miner which solves its PoW puzzle first will be offered an incentive.

\begin{figure}
	\centering   
	\includegraphics[width=2.5in]{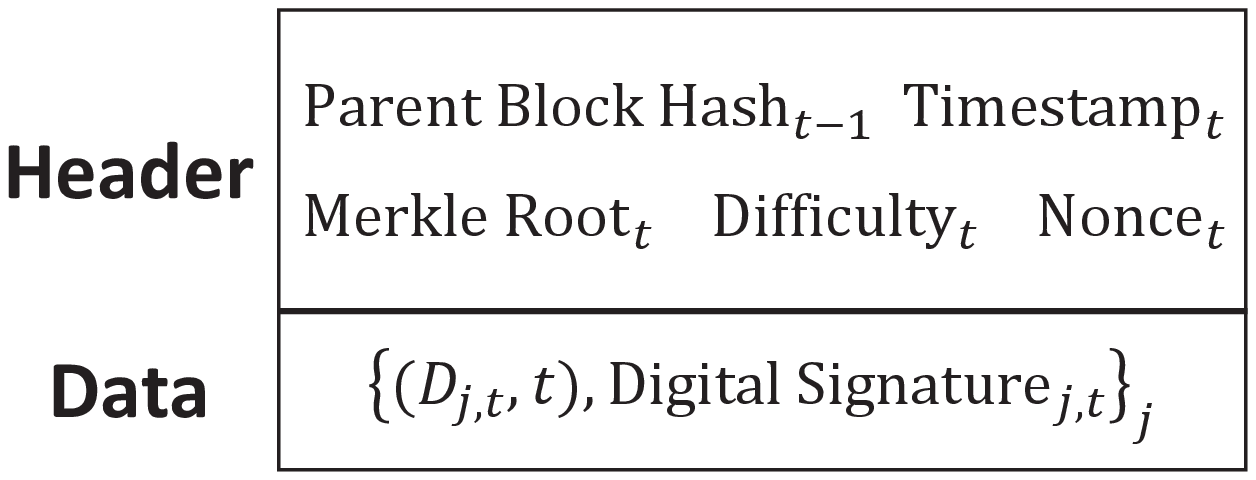}
	\caption{Structure of a block in the blockchain.}
	\label{fig_BlockStructure}
\end{figure}
The block mining process described above can be considered as a hashing competition among miners, where the probability that a miner solves its PoW puzzle first is proportional to its normalized hash rate which is defined as the ratio of its computational power to the total computational power in the network \cite{nakamoto2008bitcoin}. It is possible that two or more miners may solve their puzzles almost at the same time, which may lead to distinct blockchain branches at different miners  due to the decentralized nature of blockchain and  communication delays among nodes. However, the longest branch in each miner's  local copy of the blockchain, which is called the main chain, can reach a consensus due to the longest chain protocol of blockchain\cite{nakamoto2008bitcoin}.

The structure of the wallet's message deserves some discussion. The purpose of including a time index $t$ in every message is mainly twofold. On one hand, if there exist network latency and communication failure,  miners may lose some wallets' data packages or may not receive wallets' data packages in a chronological order. The time index included in every message can help the miners discern if there is any data package lost and rearrange the received wallet's data packages in a chronological order. On the other hand, it is possible that a wallet's data $D_{j,t}$ can be the same value for some different time indices. Including a time index $t$ in every message can distinguish the same data at different times, and also prevent potential security threats. For example, consider the case that every message does not include a time index, and there exists a malicious miner. At a time $t_1$, when forming its new block, the malicious miner can falsify the data package received from the $j$-th wallet by simply replacing it with the $j$-th wallet's data package generated at some time  $t_2<t_1$, which can be obtained in the malicious miner's copy of the blockchain. If this falsified block is mined, this falsification can pass the validation processes implemented at other nodes since the $j$-th wallet's data package generated at $t_2$ is authentic and valid. On the contrary, if every message includes a time index, then this kind of falsification cannot pass the validation processes implemented at other nodes since the time index $t_2$ contained
in the falsified  mined block is not one greater than $t_1-1$ which is the time index in the parent block of the mined block. 

As mentioned before, in a blockchain application where the task of the application has to be accomplished by some time instant, the employed blockchain can be deemed finitely-long. Let $L$ denote the number of blocks mined in the longest branch of the employed blockchain when the task of the blockchain application is accomplished. As such, in the finitely-long blockchain model, we can assume that the growth of the blockchain terminates once its longest branch grows to $L$ blocks.

\subsection{Double-Spending Attack Model}\label{Section_Adversarymodel}

In order to falsify the data which have already been stored in a blockchain, an attacker has to successfully launch a DSA \cite{nakamoto2008bitcoin}. Consider a finitely-long blockchain where $L_0$ authentic blocks which store data for the task of the application have been mined to form an authentic branch.
An attacker controls some malicious miners in the network which launches a DSA on the finitely-long blockchain in an attempt to falsify the data $\{D_{j,L_a}\}_{j\in\mathcal{A}}$ from a subset $\cA$ of wallets, which have already been stored in the $L_a$-th block $(0<L_a\le L_0)$ of the authentic branch of the blockchain, to $\{\tilde{D}_{j,L_a}\}_{j\in\mathcal{A}}$. The steps of launching the DSA on the finitely-long blockchain can be summarized as follows. The attacker first has to hack into the subset $\cA$ of wallets and steal their private keys to generate valid digital signatures for the falsified messages containing $\{\tilde{D}_{j,L_a}\}_{j\in\mathcal{A}}$. Then the attacker modifies the $L_a$-th authentic block to form a counterfeit block by replacing the $L_a$-th authentic block's data packages containing $\{D_{j,L_a}\}_{j\in\mathcal{A}}$ with the falsified data messages containing $\{\tilde{D}_{j,L_a}\}_{j\in\mathcal{A}}$ and their corresponding digital signatures. After that, the attacker works on mining this counterfeit block by redoing a PoW puzzle to find a valid nonce value for the counterfeit block which can render the hash value of the counterfeit block with no less than a prescribed number of prefix zeros. Once a valid nonce value has been successfully found, the attacker broadcasts this mined counterfeit block to the other nodes which retain local copies of the blockchain. Since the mined counterfeit block can pass the block validation conducted at the other nodes, the other nodes will add this counterfeit block after the $(L_a-1)$-th authentic block in their local copies of the blockchain. In a blockchain, only the longest branch is valid according to the longest chain protocol of blockchain. To this end, in order to ensure the validity of the $L_a$-th counterfeit block, the attacker has to mine more blocks after the $L_a$-th counterfeit block which are linked one after another to form a counterfeit branch, and extend this counterfeit branch to be the longest branch in the blockchain. The process of launching a DSA is illustrated in Fig. \ref{fig_DSA}. It is worth mentioning that mining each block of the counterfeit branch requires the malicious miners to solve a PoW puzzle. Moreover, while the malicious miners work on hacking into the subset $\cA$ of wallets and then extending the counterfeit branch, the honest miners simultaneously work on mining new blocks to extend the authentic branch before the counterfeit branch becomes the longest branch in the blockchain.

As illustrated in Fig. \ref{fig_DSA}, the counterfeit branch consists of the authentic blocks with indices from $1$ to $(L_a-1)$, the $L_a$-th counterfeit block, and all blocks linked after it. The authentic branch consists of the authentic blocks with indices from 1 to $L_0$ and all blocks linked after the $L_0$-th authentic block. It is clear that the counterfeit branch diverges from the authentic branch at the $(L_a-1)$-th block of the authentic branch. Once the counterfeit branch becomes longer than the authentic branch, all the honest miners switch from working on extending the authentic branch to working on extending the counterfeit branch according to the longest chain protocol of blockchain \cite{nakamoto2008bitcoin}. As a result, the authentic branch will stop growing, and the counterfeit branch will remain the longest branch in the blockchain as time goes by.

At last, we make a remark on the value of $L_0$ in practice. When wallets communicate with the malicious miners which are controlled by an attacker, it provides the attacker an opportunity to hack into these wallets. Generally, cybersecurity measures are taken at wallets to prevent attackers from hacking into the wallets and stealing their private keys. For example, as a basic preventative measure, most wallets are equipped with password protection to prevent hacking. In order to hack into the wallets equipped with password protection, the attacker has to employ a brute force attack which works through all possible keystrokes hoping to guess the password correctly\cite{stiawan2019investigating}. It may take a long time for the attacker to hack into all wallets in the subset $\cA$, during which  honest miners  generally can successfully mine some authentic blocks for an authentic branch, which implies that $L_0$ is generally greater than $0$ in practice owing to the cybersecurity measures  taken at  wallets.
In addition, the cases where $L_0 = L$ or $L_0=0$  are trivial. This is because if $L_0 = L$, then the task of the blockchain application has already been accomplished, and hence, the attacker has failed to launch a successful DSA. Moreover, if no authentic block has been mined in the blockchain, i.e., $L_0=0$, when the attacker accomplishes hacking into the subset $\cA$ of wallets, then no block storing the data for the task of the blockchain application has been mined in the blockchain, and therefore,  the attacker does not need to launch a DSA on the blockchain anymore.
This is because the attacker has already controlled the subset $\cA$ of wallets and can falsify their data which are transmitted to miners before the blockchain records the data for the task of the blockchain application. In light of this, we assume that $0<L_0 < L$ throughout this paper.

\begin{figure*}
	\centering   
	\includegraphics[width=1.0 \textwidth]{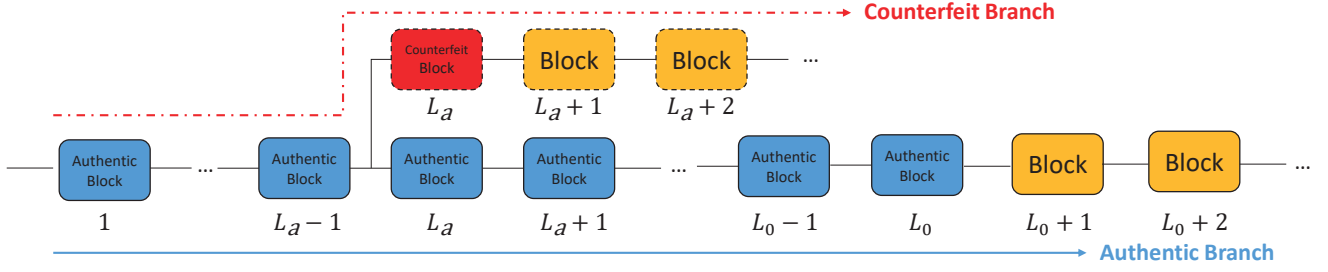}
	\caption{The finitely-long blockchain under a double-spending attack.}
	\label{fig_DSA}
\end{figure*}

\section{Probability of Success in Launching A Double-Spending Attack}\label{Section_Probability_DSA}
Consider the double-spending attack model described in Section \ref{Section_Adversarymodel} where $L_0$ blocks have been mined in an authentic branch and an attacker attempts to falsify the data stored in the $L_a$-th authentic block by launching a DSA. The attacker has to build a counterfeit branch to be longer than the authentic branch before the task of a finitely-long blockchain application is accomplished, that is, the time instant at which the authentic branch grows to $L$ blocks. 
We refer to a DSA which renders its counterfeit branch longer than the authentic branch before the authentic branch grows to $L$ blocks as a successful DSA.

According to the PoW protocol, the probability that a miner solves its PoW puzzle first among all the miners is proportional to its hash rate.  Let $I$ denote the probability that the next block mined in the blockchain is generated by a malicious miner. Hence the probability that the next block mined in the blockchain is generated by an honest miner is $1-I$. Let $\bbP_L(m,n)$ denote the probability that if the length of the authentic branch is $n$ and the counterfeit branch built by the attacker is $m$ blocks shorter than the authentic branch, then the attacker finally extends the counterfeit branch to be longer than the authentic branch before the authentic branch grows to $L$ blocks. As such,  $\bbP_L(L_0-L_a+1,L_0)$  is the probability that the attacker launches a successful DSA.

Note that if a malicious miner mines the next block in the blockchain, which happens with probability $I$,  the counterfeit branch will be  $(m-1)$ blocks shorter than the authentic branch while the authentic branch remains $n$ blocks. In contrast, if an honest miner mines the next block in the blockchain, which happens with probability $(1-I)$, the counterfeit branch will be $(m+1)$ blocks shorter than the authentic branch and the authentic branch grows to $(n+1)$ blocks. To this end, 
the recursive form of $\bbP_L(m,n)$ can be expressed as,  $\forall m \in \{0,\cdots, L_0-L_a+1\}$ and  $\forall n \in \{L_0,\cdots, L\}$, 
\begin{equation} \label{equ_2D}
\bbP_L(m,n)=I\times \bbP_L(m-1,n)+(1-I)\times \bbP_L(m+1,n+1).
\end{equation}
 
Once the counterfeit branch becomes longer than the authentic branch (i.e., the counterfeit branch is $-1$ blocks shorter than the authentic branch) before the length of the authentic branch reaches $L$, the DSA launched by the attacker succeeds, and therefore, we have the following boundary condition
\begin{equation} \label{boundary_condition_1}
	\bbP_L(-1,n)=1, \quad \forall 0<n<L.
\end{equation}
On the other hand, if the length of the authentic branch grows to $L$ and the counterfeit branch is still shorter than the authentic branch, the blockchain reaches its final length which implies that the attacker fails to launch a successful DSA. Hence, we have another boundary condition
\begin{equation}\label{boundary_condition_2}
	\bbP_L(m,L)=0, \quad \forall m>0.
\end{equation}
It is worth pointing out that when an honest miner mines a block and then the length of the authentic branch grows to $L$, the length of the counterfeit branch cannot be the same as that of the authentic branch. This is because if so, then before the honest miner mines a block, the counterfeit branch has already grown to $L$ blocks, and hence, the counterfeit branch has already been longer than the authentic branch which implies that we have already reached the boundary condition in (\ref{boundary_condition_1}) before the honest miner mines a block. To this end, we don't need to include the case where $m=0$ in the boundary condition in (\ref{boundary_condition_2}).
\begin{figure}\centering 
	\includegraphics[width=0.45\textwidth]{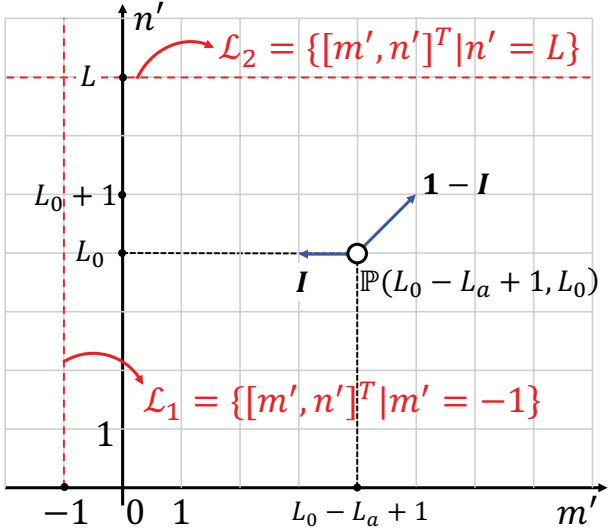}
	\caption{Two-dimensional random walk illustration.}
	\label{Fig_random_walk}
\end{figure}

By employing (\ref{equ_2D}), (\ref{boundary_condition_1}) and (\ref{boundary_condition_2}), the pursuit of the closed-form expression for $\bbP_L(m,n)$  can be cast as a two-sided boundary hitting problem for a two-dimensional random walk with two possible moving directions, which is illustrated in Fig. \ref{Fig_random_walk}. To be specific, let ${\boldsymbol s}_T \triangleq [m,n]^T + \sum_{t=1}^{t=T} \bdelta_t$ denote the position of the random walk in two-dimensional space after the $T$-th step (i.e., after $T$ blocks have been mined in the blockchain), where for any $t$, $\bdelta_t$ is a random vector 
\begin{equation}
	{\bdelta _t} = \left\{ {\begin{array}{*{20}{c}}
{\bdelta^{(1)} \triangleq  {{[ - 1,0]}^T}}&{\text{with probability }I,}\\
{\bdelta^{(2)} \triangleq {{[1,1]}^T}}&{\text{with probability } 1 - I.}
\end{array}} \right.
\end{equation}
Let $\cL_1  \buildrel \Delta \over = \{ {[m',n']^T}\left| {m' =  - 1} \right.\}$ and ${\cL_2} \buildrel \Delta \over = \{ {[m',n']^T}\left| {n' = L} \right.\}$ define two boundary lines in  two-dimensional space, respectively.  As such, $\bbP_L(m,n)$ can be rewritten  as
\begin{equation} \label{P_m_n_closed_form}
	\bbP_L(m,n) = \sum_{T = 0}^\infty  {\Pr \left( {{\bs_T} \in {\cL_1},{\bs_{T'}} \notin {\cL_1} \cup {\cL_2},\;\forall T' < T} \right)}
\end{equation}

The closed-form expression for $\bbP_L(m,n)$ for any given $m,n$ is described in the following theorem.
\begin{theorem}\label{Theorem_DSA}
When the length of the authentic branch is $n$ and the counterfeit branch is $m$ blocks shorter than the authentic branch, the probability $\bbP_L(m,n)$ that an attacker launches a successful DSA can be expressed as
\begin{equation}\setlength{\arraycolsep}{1pt}
\begin{aligned} \label{P_m_n_theorem}
&\bbP_L(m,n)\\
&=\left\{\begin{array}{{c}{l}}\sum_{i=0}^{L-n-1}a_{i,m}(1-I)^iI^{m+1+i}&,\begin{array}{l}\text{if }m\ge0,0<n<L,\\
\text{and }0\le I<1\end{array}\\
1&,\begin{array}{l}\text{if }m=-1,0\!<\!n\!<\!L,\\
\text{or }I=1,0<n<L\end{array}\\
0&,\text{if }m>0,n=L\end{array}\right.
\end{aligned}
\end{equation}
where the coefficients $a_{i,m}$ is defined as
\begin{equation}\label{a_i_m_theorem}
\setlength{\arraycolsep}{1pt}
{a_{i,m}} = \left\{ {\begin{array}{*{5}{cl}}
	{1,}&{{\rm{if }} \; i = 0,}\\
	{1 + m,}&{{\rm{if }}\; i = 1,}\\
	{{C_i},}&{{\rm{if }}\; m = 0,}\\
	{{C_{i + 1}},}&{{\rm{if }}\; m = 1,}\\
	\begin{array}{l}
	{C_{i + 1}} + \sum\limits_{{j_1} = 3}^{m + 1} {{C_i}} + \sum\limits_{{j_1} = 3}^{m + 1} {\sum\limits_{{j_2} = 3}^{{j_1} + 1} {{C_{i - 1}}} }\\
 \quad    +  \cdots  
	   + \sum\limits_{{j_1} = 3}^{m + 1} {\sum\limits_{{j_2} = 3}^{{j_1} + 1}  \cdots  } \sum\limits_{{j_{i - 2}} = 3}^{{j_{i - 3}} + 1} {{C_3}} \\
	\quad   + \sum\limits_{{j_1} = 3}^{m + 1} {\sum\limits_{{j_2} = 3}^{{j_1} + 1}  \cdots  } \sum\limits_{{j_{i - 1}} = 3}^{{j_{i - 2}} + 1} {(1 + {j_{i-1}})} ,
	\end{array}&\ \begin{aligned}&{\rm{if }} \; i \! > \!1\\&\rm{and}\;m > 1\end{aligned},
	\end{array}} \right.
\end{equation} 	
and the constant $C_i$ is the $i$-th Catalan number which is given by
\begin{equation}
	C_i \buildrel \Delta \over =  \frac{1}{i+1}\left(\begin{matrix}2i\\i\end{matrix}\right) = \frac{{(2i)!}}{{(i + 1)!i!}}.
\end{equation}
\end{theorem}
\begin{proof}[Proof:\nopunct]
From (\ref{equ_2D}) and (\ref{boundary_condition_1}), we can easily see that $\bbP_L(m,n)=1$ if $m=-1$, or if $I=1$ and $0<n<L$, $\bbP_L(m,n)=0$ if $n=L$. Next, we prove the case that $0\le I<1$ and $0<n<L$.

	As illustrated in Fig. \ref{Fig_random_walk}, if the random walk can hit the boundary line $\cL_1$ before hitting the boundary line $\cL_2$, all the possible cases are that $\forall i=0,1,...,L-n-1$, the random walk moves along $\bdelta^{(1)}$ by $(m+1+i)$ steps and moves along $\bdelta^{(2)}$ by $i$ steps.
Thus, $\bbP_L(m,n)$ can be written in the general form
\begin{equation} \label{equ_genForm}
\bbP_L(m,n)=\sum_{i=0}^{L-n-1}a_{i,m}(1-I)^iI^{m+1+i}
\end{equation}
where $a_{i,m}$ is a constant which denotes the number of all the paths that the random walk hits the boundary $\cL_1$ in the $(m+2i+1)$-th step by moving $(m+i+1)$ steps along $\bdelta^{(1)}$  and $i$ steps along $\bdelta^{(2)}$ and the random walk never hits $\cL_1$ before the $(m+2i+1)$-th step.  It is obvious that 
\begin{equation} \label{a_0_m}
	a_{0,m}=1 \;\; \text{for  } m \ge 0,
\end{equation}
and hence
\begin{equation} \label{P_m_n_alt}
	\bbP_L(m,n)=I^{(m+1)} + \sum_{i=1}^{L-n-1}a_{i,m}(1-I)^iI^{m+1+i}.
\end{equation}
By substituting (\ref{equ_genForm}) and (\ref{P_m_n_alt}) into (\ref{equ_2D}), we can obtain
\begin{align} \notag
&I^{m+1}+\sum_{i=1}^{L-n-1}a_{i,m}(1-I)^iI^{m+1+i}\\ \notag
&=I\times\left(\sum_{i=0}^{L-n-1}a_{i,m-1}(1-I)^iI^{m+i}\right) \\ \notag
&\qquad  +(1-I)\times \left(\sum_{i=0}^{L-n-2}a_{i,m+1}(1-I)^iI^{m+2+i}\right)\\  \notag
&=\sum_{i=0}^{L-n-1}a_{i,m-1}(1-I)^iI^{m+1+i}\\\notag
&\qquad+\sum_{i=0}^{L-n-2}a_{i,m+1}(1-I)^{i+1}I^{m+2+i}\\  \notag
&=I^{m+1}+\sum_{i=1}^{L-n-1}a_{i,m-1}(1-I)^iI^{m+1+i}\\  \notag
&\qquad  +\sum_{i=1}^{L-n-1}a_{i-1,m+1}(1-I)^{i}I^{m+1+i}\\ \label{I_poly}
&=I^{m+1}+\sum_{i=1}^{L-n-1}[a_{i,m-1}+a_{i-1,m+1}](1-I)^iI^{m+1+i}.
\end{align}
Since (\ref{I_poly}) holds for any $I$, by the fundamental theorem of algebra, we know 
\begin{equation} \label{equ_Recur}
a_{i,m}=a_{i,m-1}+a_{i-1,m+1}, \;\; \forall i>0 \text{ and }m\ge0.
\end{equation}
Since $a_{0,m}=1$ for any $m \ge 0$, we can get \begin{equation} \label{a_1_m}
	a_{1,m}=a_{1,m-1}+1 = a_{1,0} + m ,\;\; \forall m\ge 0
\end{equation}
from (\ref{equ_Recur}) by choosing $i=1$.
Furthermore, by setting $m=0$ in (\ref{equ_2D}), we have 
\begin{equation}
	\bbP_L(0,n)=I+(1-I)\times \bbP_L(1,n+1)
\end{equation}
which yields 
\begin{equation}\label{I_poly_2}
I+\sum_{i=1}^{L-n-1}a_{i,0}(1-I)^{i}I^{1+i}=I+\sum_{i=1}^{L-n-1}a_{i-1,1}(1-I)^{i}I^{1+i}
\end{equation}
by employing (\ref{equ_genForm}).

By the fundamental theorem of algebra, we know from (\ref{I_poly_2}) that
\begin{equation} \label{a_recursive}
a_{i,0}=a_{i-1,1}, \;\; \forall i=1,...,L-n-1.
\end{equation}
From (\ref{a_0_m}) and (\ref{a_recursive}), we can obtain
\begin{equation}
	a_{1,0}=a_{0,1}=1,
\end{equation}  
and therefore, 
\begin{equation} \label{a_1_m_final}
	a_{1,m}=1+m \;\; \forall m\ge 0
\end{equation}
by employing (\ref{a_1_m}).
From the recursive equation in (\ref{equ_Recur}), we can obtain that for any $i =1,2,...,L-n-1$,
\begin{equation} \label{a_i_m_1}
	a_{i,m}=a_{i,1}+\sum_{j=3}^{m+1}a_{i-1,j}, \;\forall m>1.
\end{equation}
Moreover, from (\ref{a_recursive}), we know
\begin{equation} \label{a_i_1_2}
	a_{i,1}=a_{i+1,0}, \;\; \forall i=0,...,L-n-2,
\end{equation} 
which implies that for all $i=1,...,L-n-2$ and $m > 1$,
\begin{align} \notag
a_{i,m}&=a_{i+1,0}+\sum_{j_1=3}^{m+1}a_{i-1,j_1}\\  \notag
&=a_{i+1,0}+\sum_{j_1=3}^{m+1}a_{i,0}+\sum_{j_1=3}^{m+1}\sum_{j_2=3}^{j_1+1}a_{i-1,0}+\cdots\\ \notag
&\quad+\sum_{j_1=3}^{m+1}\sum_{j_{2}=3}^{j_{2}+1}\cdots\sum_{j_{i-2}=3}^{j_{i-3}+1}a_{3,0}+\sum_{j_1=3}^{m+1}\sum_{j_{2}=3}^{j_{2}+1}\cdots\sum_{j_{i-1}=3}^{j_{i-2}+1}a_{1,j_{i-1}} \\ \notag
&=a_{i+1,0}+\sum_{j_1=3}^{m+1}a_{i,0}+\sum_{j_1=3}^{m+1}\sum_{j_2=3}^{j_1+1}a_{i-1,0}+\cdots\\ \notag
&\quad+\sum_{j_1=3}^{m+1}\sum_{j_{2}=3}^{j_{2}+1}\cdots\sum_{j_{i-2}=3}^{j_{i-3}+1}a_{3,0} \\  \label{a_i_m_final_1}
& \quad +\sum_{j_1=3}^{m+1}\sum_{j_{2}=3}^{j_{2}+1}\cdots\sum_{j_{i-1}=3}^{j_{i-2}+1}(1+j_{i-1}),
\end{align}
by employing (\ref{a_1_m_final}) and (\ref{a_i_m_1}).

Next, we need to determine $a_{i,m}$ when $i=L-n-1$.
Suppose that there is another boundary line ${\cL_3} \buildrel \Delta \over = \{ {[m',n']^T}\left| {n' = L+1} \right.\}$ and $\bbP_L'(m,n)$ is the probability that the random walk hits the boundary line $\cL_1$ before hitting the boundary line $\cL_3$. Similar to (\ref{equ_genForm}), we know 
$$\bbP_L'(m,n)=\sum_{i=0}^{L-n}a'_{i,m}(1-I)^iI^{m+1+i}$$
where $a'_{i,m}$ is a constant which denotes the number of all the paths that the random walk hits the boundary $\cL_1$ in the $(m+2i+1)$-th step by moving $(m+1+i)$ steps along $\bdelta^{(1)}$  and $i$ steps along $\bdelta^{(2)}$ and the random walk never hits $\cL_1$  before the $(m+2i+1)$-th step. It is clear that 
\begin{equation} \label{a_prime_id}
	a'_{i,m}=a_{i,m}, \;\; \forall i=0,...,L-n-1 \text{ and }\forall m\ge 0.
\end{equation}
Similar to (\ref{a_i_1_2}), we can get 
\begin{equation} \label{a_prime_rec}
	a'_{i,1}=a'_{i+1,0}, \;\; \forall i=0,...,L-n-1.
\end{equation}
From (\ref{a_prime_id}) and (\ref{a_prime_rec}), we know
\begin{equation}
a_{L-n-1,1}=a'_{L-n-1,1}=a'_{L-n,0}.
\end{equation}
Similar to (\ref{a_i_m_final_1}), by employing (\ref{a_prime_id}), we can obtain that for all $m > 1$,
\begin{align} \notag
a_{L-n-1,m} & = a'_{L-n-1,m} \\ \notag
&=a'_{L-n,0}+\sum_{j_1=3}^{m+1}a_{L-n-2,j_1}\\ \notag
&=a'_{L-n,0}+\sum_{j_1=3}^{m+1}a_{L-n-1,0}+\sum_{j_1=3}^{m+1}\sum_{j_2=3}^{j_1+1}a_{L-n-2,0}\\ \notag
&\qquad+\cdots+\sum_{j_1=3}^{m+1}\sum_{j_{2}=3}^{j_{2}+1}\cdots\sum_{j_{L-n-3}=3}^{j_{L-n-4}+1}a_{3,0} \\ \label{a_i_m_final_2}
& \qquad +\sum_{j_1=3}^{m+1}\sum_{j_{2}=3}^{j_{2}+1}\cdots\sum_{j_{L-n-2}=3}^{j_{L-n-3}+1}(1+j_{L-n-2}).
\end{align}

It is seen from (\ref{a_0_m}), (\ref{a_i_m_final_1}), (\ref{a_prime_id}) and (\ref{a_i_m_final_2}) that if we can determine the value of $a'_{i,0}$ for all $i=3,4,...L-n$, then we can determine the value of $a_{i,m}$ for all $i=0,1,...L-n-1$ and $\forall m \ge 0$, and hence, we can obtain the closed-form expression for $\bbP_L(m,n)$ by using (\ref{equ_genForm}). In what follows, we will derive the closed-form expression for $a'_{i,0}$ for all $i=3,4,...L-n$. 



Note that  $a'_{i,0}$ denotes the number of all the paths that the random walk hits the boundary line $\cL_1$ in the $(2i+1)$-th step by moving $(i+1)$ steps along $\bdelta^{(1)}$  and $i$ steps along $\bdelta^{(2)}$ and the random walk never hits $\cL_1$ before the $(2i+1)$-th step. We know that for any such path in two-dimensional space, it starts from the point $[0,n]^T$ and arrives at the point $[0,n+i]^T$  in the $(2i)$-th step, and moreover, the first $2i$ steps of the path must stay in the half space $\cS \triangleq  \{ {[m',n']^T}\left| {m' \ge 0} \right.\}$.
In light of this, $a'_{i,0}$ is identical to the number of all the lattice paths from the point $[0,0]^T$ to the point $[i,i]^T$ which consist of $i$  steps along the vector $[0,1]^T$ and $i$  steps along the vector $[1,0]^T$ and never rise above the diagonal line in the $i$-by-$i$ grid. 
Such paths are referred to as the Dyck Paths, and the number of such paths is known as the $i$-th Catalan number $C_i$ \cite{stanley2015catalan}. Therefore, we know
\begin{equation} \label{a_prime_c_i}
	a'_{i,0}=C_i \triangleq  \frac{1}{i+1}\left(\begin{matrix}2i\\i\end{matrix}\right) = \frac{{(2i)!}}{{(i + 1)!i!}}, \;\; \forall i=3,...,L-n-1.
\end{equation}
By employing (\ref{boundary_condition_1}), (\ref{boundary_condition_2}), (\ref{equ_genForm}), (\ref{a_0_m}), (\ref{a_1_m_final}), (\ref{a_i_m_final_1}), (\ref{a_prime_rec}),  (\ref{a_i_m_final_2}) and (\ref{a_prime_c_i}), we can obtain
\begin{equation}\setlength{\arraycolsep}{1pt}
\begin{aligned}
&\bbP_L(m,n)\\
&=\left\{\begin{array}{{c}{l}}\sum_{i=0}^{L-n-1}a_{i,m}(1-I)^iI^{m+1+i}&,\begin{array}{l}\text{if }m\ge0,0<n<L,\\
\text{and }0\le I<1\end{array}\\
1&,\begin{array}{l}\text{if }m=-1,0\!<\!n\!<\!L,\\
\text{or }I=1,0<n<L\end{array}\\
0&,\text{if }m>0,n=L\end{array}\right.
\end{aligned}
\end{equation}
where the coefficients $a_{i,m}$ can be expressed as

\begin{equation}
\setlength{\arraycolsep}{1pt}
{a_{i,m}} = \left\{ {\begin{array}{*{5}{c}}
	{1,}&{{\rm{if }} \; i = 0,}\\
	{1 + m,}&{{\rm{if }}\; i = 1,}\\
	{{C_i},}&{{\rm{if }}\; m = 0,}\\
	{{C_{i + 1}},}&{{\rm{if }}\; m = 1,}\\
	\begin{array}{l}
	{C_{i + 1}} + \sum\limits_{{j_1} = 3}^{m + 1} {{C_i}} + \sum\limits_{{j_1} = 3}^{m + 1} {\sum\limits_{{j_2} = 3}^{{j_1} + 1} {{C_{i - 1}}} }\\
 \quad   +  \cdots  
	   + \sum\limits_{{j_1} = 3}^{m + 1} {\sum\limits_{{j_2} = 3}^{{j_1} + 1}  \cdots  } \sum\limits_{{j_{i - 2}} = 3}^{{j_{i - 3}} + 1} {{C_3}} \\
	\quad   + \sum\limits_{{j_1} = 3}^{m + 1} {\sum\limits_{{j_2} = 3}^{{j_1} + 1}  \cdots  } \sum\limits_{{j_{i - 1}} = 3}^{{j_{i - 2}} + 1} {(1 + {j_{i-1}})} ,
	\end{array}&\ \begin{aligned}&{\rm{if }} \; i  > 1\\&\rm{and}\;m > 1\end{aligned},
	\end{array}} \right.
\end{equation} 	
which completes the proof.

\end{proof}

By substituting $L_0-L_a+1$ and $L_0$ for $m$ and $n$ in the general closed-form expression in (\ref{P_m_n_theorem}), respectively, we can obtain the probability that the attacker launches a successful DSA on the finitely-long blockchain. Intuitively speaking, in the asymptotic regime where $L\to\infty$, the boundary line $\cL_2$ in Fig. \ref{Fig_random_walk} essentially does not exist, and hence, the two-dimensional random walk boundary hitting problem illustrated in (\ref{P_m_n_closed_form}) degenerates to a one-dimensional random walk boundary hitting problem. In consequence, the probability that the attacker launches a successful DSA on a finitely-long blockchain is expected to converge to that on an infinitely-long blockchain as $L\to\infty$. This intuitive conjecture is formally summarized and rigorously proved in the following theorem.

\begin{theorem}\label{Theorem_asymptotic}
If $0<n<L$, then \cbk$\bbP_L(m,n)$ strictly increases as $L$ increases when $0<I<1$,  
 and is constant when $I=0$ or $1$. Moreover, when $0<n<L$, as $L \to \infty$, the probability $\bbP_L(m,n)$ in (\ref{P_m_n_closed_form}) converges to
\begin{equation} \label{P_m_n_asymptotic}
\lim_{L\to\infty}\bbP_L(m,n)=\left\{{\begin{array}{*{5}{c}} 
\left(\frac{I}{1-I}\right)^{m+1}&,\text{ if }\;0\le I<0.5,\\
1&,\text{ if }\;0.5\le I\le 1.
\end{array}}\right.
\end{equation}
which equals the probability that an attacker launches a successful DSA on an infinitely-long blockchain when the counterfeit branch is $m$ blocks shorter than the authentic branch  \cite{nakamoto2008bitcoin,grunspan2018double,zaghloul2020bitcoin} .
\end{theorem}
\begin{proof}[Proof:\nopunct]
When $0<n<L$, from (\ref{P_m_n_theorem}), we have
\begin{align}\notag
\bbP_{L+1}(m,n)&=\sum_{i=0}^{L-n}a_{i,m}(1-I)^iI^{m+1+i}\\
&=\bbP_L(m,n)+a_{{L-n},m}(1-I)^{L-n}I^{m+1+L-n}.
\end{align}
From (\ref{a_i_m_theorem}), we know that $a_{i,m}\ge1$,$\forall i,m\ge0$, and hence, we have $a_{{L-n},m}(1-I)^{L-n}I^{m+1+L-n}>0$ when $0<I<1$ and $a_{{L-n},m}(1-I)^{L-n}I^{m+1+L-n}=0$ when $I=0$ or $1$. This implies that 
\begin{align}\label{P_increasing}
\bbP_{L+1}(m,n)>\bbP_{L}(m,n)&, \;\; \text{if }  0<I<1,\\ \label{P_constant}
\bbP_{L+1}(m,n)=\bbP_{L}(m,n)&, \;\; \text{if } I=0 \text{ or } 1.
\end{align}
From (\ref{P_increasing}) and (\ref{P_constant}), we know that when $0<I<1$,  $\bbP_{L}(m,n)$ is a strictly increasing function of $L$,  and when $I=0$ or $1$, $\bbP_{L}(m,n)$ is a constant as $L$ changes. 
By the definition of $\bbP_L(m,n)$, we know that $\bbP_L(m,n)\le1$, $\forall L$. Since when $0<I<1$, $\bbP_{L}(m,n)$ is strictly increasing and bounded from above, we know $\bbP_{L}(m,n)$ converges as $L$  increases by the monotone convergence theorem \cite{wheeden1977measure}. When $I=0$ or $1$, $\bbP_{L}(m,n)$ doesn't change as $L$ changes, and hence, $\bbP_{L}(m,n)$ also converges. 
Therefore, $\bbP_{L}(m,n)$ converges as $L$ increases, which implies
\begin{equation}\label{p_infty}
\lim_{L\to\infty}\bbP_{L}(m,n)=\lim_{L\to\infty}\bbP_{L+1}(m,n).\end{equation}
Note that if $n\ge L_0$, we have
\begin{align}\label{p_l+1} \notag
\bbP_{L+1}&(m+1,n+1)\\ \notag
&=\sum_{i=0}^{(L+1)-(n+1)-1}a_{i,m+1}(1-I)^iI^{(m+1)+1+i}\\ \notag
&=\sum_{i=0}^{(L)-(n)-1}a_{i,m+1}(1-I)^iI^{(m+1)+1+i}\\
&=\bbP_{L}(m+1,n).
\end{align}
From (\ref{p_infty}) and (\ref{p_l+1}), we can get  
\begin{equation}\label{p_n-1}
\lim_{L\to\infty}\bbP_{L}(m+1,n+1)=\lim_{L\to\infty}\bbP_{L}(m+1,n), \;\; \text{if }n\ge L_0.
\end{equation}
From (\ref{equ_2D}), we can get 
\begin{align}\label{equ_2d_infty}\notag
\lim_{L\to\infty}\bbP_L(m,n)=&\ I\times \lim_{L\to\infty}\bbP_L(m-1,n)\\
&+(1-I)\times \lim_{L\to\infty}\bbP_L(m+1,n+1).
\end{align}
Note that the condition $n\ge L_0$ holds for the cases of our interest since the authentic branch starts with $L_0$ blocks when the attacker launches a DSA. Substituting (\ref{p_n-1}) into (\ref{equ_2d_infty}), we have
\begin{align}\label{equ_1d_infty}\notag
\lim_{L\to\infty}\bbP_L(m,n)=&\ I\times \lim_{L\to\infty}\bbP_L(m-1,n)\\
&+(1-I)\times \lim_{L\to\infty}\bbP_L(m+1,n).
\end{align}
For simplicity, let $g(m+1)\triangleq \lim_{L\to\infty}\bbP_L(m,n)$, $m\ge -1$,  then from (\ref{equ_1d_infty}), we can obtain
\begin{equation}\label{equ_1d_g}
g(m+1)=I\times g(m)+(1-I)\times g(m+2).
\end{equation}
From (\ref{boundary_condition_1}), we have $\lim_{L\to\infty}\bbP_L(-1,n)=1$ which implies 
\begin{equation}\label{g_0}
g(0)=\lim_{L\to\infty}\bbP_L(-1,n)=1.
\end{equation}
From (\ref{equ_1d_g}), we have
\begin{equation}\label{equ_1d_grecursive}
g(m+2)-g(m+1)=\frac{I}{1-I}\left[g(m+1)-g(m)\right],
\end{equation}
which yields
 \begin{equation} \label{equ_1d_grecursive1}
g(m+1)-g(m)=\left(\frac{I}{1-I}\right)^m\left[g(1)-g(0)\right].
\end{equation}
Moreover, from the recursive equation in (\ref{equ_1d_grecursive1}), we can obtain
 \begin{align} \label{equ_1d_grecursive2}
g(m+1)-g(1)=\left[g(1)-g(0)\right]\sum_{i=1}^m\left(\frac{I}{1-I}\right)^m.
\end{align}
Next, we will determine $g(1)$. From (\ref{P_m_n_theorem}) and (\ref{a_i_m_theorem}), we know that for any $n$,
\begin{align}\label{g_1}\notag
g(1)=&\lim_{L\to\infty}\bbP_L(0,n)\\ \notag
=&\lim_{L\to\infty}\sum_{i=0}^{L-n-1}C_{i}(1-I)^iI^{1+i} \\ 
= &\sum_{i=0}^{\infty}C_{i}(1-I)^iI^{1+i}.
\end{align}
where $C_i$ is the $i$-th Catalan number. 
Note that the generating function for the Catalan numbers is defined by \cite{stanley2015catalan}
\begin{align}\label{generating_catalan}
c(x)=\sum_{i=0}^{\infty}C_{i}x^{i}=\frac{1-\sqrt{1-4x}}{2x}.
\end{align}
By setting $x=I(1-I)$ for some $I\in (0,1)$ in (\ref{generating_catalan}), we can obtain 
\begin{align}\label{generating_I} \notag
\sum_{i=0}^{\infty}C_{i}\left[I(1-I)\right]^{i}=&\frac{1-\sqrt{1-4I(1-I)}}{2I(1-I)}\\ \notag
=&\frac{1-\sqrt{(1-2I)^2}}{2I(1-I)}\\ 
=&\left\{\begin{array}{cc}\frac{1}{1-I}&,0<I<0.5,\\\frac{1}{I}&,0.5\le I<1.\end{array}\right.
\end{align}
Note that if $I=0$, then $g(1)=0$, and if $I=1$, then $g(1)=1$. Moreover, since $I\times\sum_{i=0}^{\infty}C_{i}(I(1-I))^{i}=\sum_{i=0}^{\infty}C_{i}(1-I)^iI^{1+i}=g(1)$, 
we can obtain from (\ref{generating_I})
\begin{align}\label{equ_g1} 
g(1)=\left\{\begin{array}{cc}\frac{I}{1-I}&,I<0.5,\\1&,I\ge0.5.\end{array}\right.
\end{align}
By employing (\ref{g_0}), (\ref{equ_1d_grecursive2}) and (\ref{equ_g1}), we can obtain
\begin{equation}
\lim_{L\to\infty}\bbP_L(m,n)=\left\{\begin{array}{*{5}{c}}(\frac{I}{1-I})^{m+1}&,\text{if }0\le I<0.5,\\1&,\text{if }0.5\le I\le 1.\end{array}\right.
\end{equation}
which completes the proof.
\end{proof}
In Theorem \ref{Theorem_asymptotic}, the assumption $0<n<L$ holds for the cases of our interest since $0<L_0<L$.
As demonstrated by Theorem \ref{Theorem_asymptotic}, since $\bbP_L(m,n)$ is an increasing function of $L$, and as $L$ increases, it converges to the probability of success in launching a DSA on an infinitely-long blockchain, we know that  the probability of success in launching a DSA on an infinitely-long blockchain when the counterfeit branch is $m$ blocks shorter than the authentic branch is an upper bound on the probability of success in launching a DSA on a finitely-long blockchain for any $m$ and $n$. It is well known that if an attacker owns more than half of the computational power of the whole network (i.e., $I>0.5$) and launches a DSA on an infinitely-long blockchain, which is called a 51\% attack, the probability of success in launching the 51\% attack is always one \cite{nakamoto2008bitcoin}. However, as shown by Theorem \ref{Theorem_asymptotic}, for the case where $0<I<1$ which is generally true if the normalized hash rate of the attacker is between 0 and 1, $\bbP_L(m,n)$ is a strictly increasing function of $L$, and therefore, we can see from (\ref{P_m_n_asymptotic}) that $\bbP_L(m,n)$ is strictly smaller than 1 even though $I>0.5$. This demonstrates that unlike infinitely-long blockchains, the probability of success in launching a 51\% attack on any finitely-long blockchain is strictly smaller than 1 which reveals that finitely-long  blockchains are more resistant to 51\% attacks when compared with infinitely-long blockchains.

\section{Simulation Results}
\label{Section_simulation}

In this section, we numerically study the probability $\bbP_L(L_0-L_a+1,L_0)$ that an attacker launches a successful DSA on a finitely-long blockchain for different parameters $L_0$, $L_a$, $L$, and $I$, respectively.  In particular, we employ Monte Carlo simulations to corroborate the theories developed in this paper. The number of Monte Carlo runs is $10^4$ in all simulation results, and the Monte Carlo simulation results are specified by the legend label ``Simulation'' in figures.

As the index $L_a$ of the block that an attacker attempts to falsify varies from $1$ to $60$, Fig. \ref{fig_PI} depicts $\bbP_L(L_0-L_a+1,L_0)$ for different $I$, where $L_0$ and $L$ are chosen to be 60 and 100, respectively. The numerical results yielded from Theorem \ref{Theorem_DSA} and Monte Carlo simulations are specified by dashed and solid lines in Fig. \ref{fig_PI}, respectively, which clearly agree with each other. Hence, the numerical results in Fig. \ref{fig_PI} corroborate Theorem \ref{Theorem_DSA}.  It is seen from Fig. \ref{fig_PI} that $\bbP_L(L_0-L_a+1,L_0)$ increases as $L_a$ increases, and moreover, for a given $L_a$, $\bbP_L(L_0-L_a+1,L_0)$ increases as $I$ increases.  This is because when $L_a$ increases, the gap between the number of blocks in the counterfeit branch and that in the authentic branch shrinks, 
and therefore, it is easier for the attacker to extend its counterfeit branch to surpass the authentic branch before the authentic branch grows to $L$ blocks. Moreover, when $I$ increases, the probability that a malicious miner successfully mines the next block in the network increases, and hence, it is also easier for the attacker to launch a successful DSA.
\begin{figure}[H]\centering   
	\includegraphics[width=0.5\textwidth]{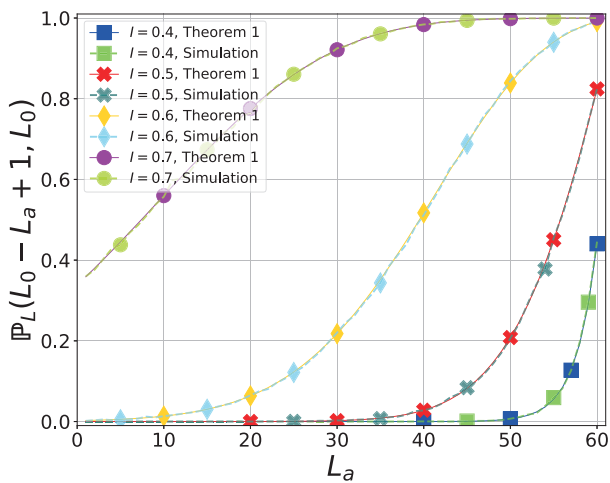}    
	\caption{The relationship between $\bbP_L(L_0-L_a+1,L_0)$ and $L_a$ for different $I$. }
	\label{fig_PI}
\end{figure}
Next, we numerically investigate how the value of $L$ impacts $\bbP_L(L_0-L_a+1,L_0)$. As $L_a$ varies from $1$ to $10$, Fig. \ref{fig_PL} depicts $\bbP_L(L_0-L_a+1,L_0)$ for different $L$, where $L_0=10$ and $I=0.4$. It is seen from Fig. \ref{fig_PL} that the numerical results yielded from Theorem 1 agree with that from Monte Carlo simulations. Moreover, for a given $L_a$, $\bbP_L(L_0-L_a+1,L_0)$ increases as $L$ increases. This can be explained by the fact that when $L$ increases, it takes more time for the honest miners to extend the authentic branch to reach $L$ blocks, and hence, the attacker has a better chance to extend the counterfeit branch to be longer than the authentic branch before the authentic branch grows to $L$ blocks.

\begin{figure}[H]\centering   
	\includegraphics[width=0.5\textwidth]{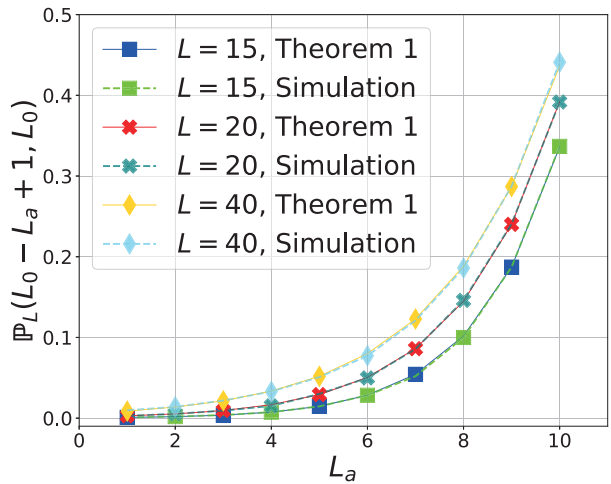}    
	\caption{The relationship between $\bbP_L(L_0-L_a+1,L_0)$ and $L_a$ for different $L$.}
	\label{fig_PL}
\end{figure}

\begin{figure}[H]\centering   
	\includegraphics[width=0.5\textwidth]{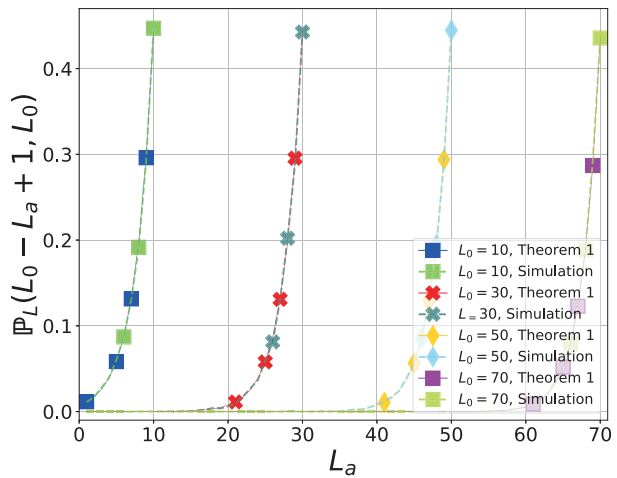}    
	\caption{The relationship between $\bbP_L(L_0-L_a+1,L_0)$ and $L_a$ for different $L_0$.}
	\label{fig_PN}
\end{figure}

 Fig. \ref{fig_PN} depicts $\bbP_L(L_0-L_a+1,L_0)$ for different $L_0$ as $L_a$ varies from $1$ to $70$, where $L=100$ and $I=0.4$. Again, the numerical results yielded from Theorem \ref{Theorem_DSA} exactly match the Monte Carlo simulation results. It is seen from Fig. \ref{fig_PN} that for a given $L_a$, $\bbP_L(L_0-L_a+1,L_0)$ increases as $L_0$ decreases. This is because when $L_a$ is given and $L_0$ decreases, the gap between the number of blocks in the counterfeit branch and that in the authentic branch shrinks, and hence, the attacker has a better chance to extend the counterfeit branch to be longer than the authentic branch before the authentic branch grows to $L$ blocks.

Lastly, we numerically corroborate Theorem \ref{Theorem_asymptotic}. As $L$ increases, Fig. \ref{fig_vs} depicts $\bbP_L(L_0-L_a+1,L_0)$ obtained from (\ref{P_m_n_theorem}) and (\ref{P_m_n_asymptotic}) for different $I$, where $L_a=1$ and $L_0=3$. The blue and red curves  are obtained from (\ref{P_m_n_theorem}) and (\ref{P_m_n_asymptotic}), respectively, when $I=0.4$. The yellow  and purple curves  are obtained from (\ref{P_m_n_theorem}) and (\ref{P_m_n_asymptotic}), respectively, when $I=0.6$. It is seen from Fig. \ref{fig_vs} that the curves obtained from (\ref{P_m_n_theorem}) converge to the   corresponding curves obtained from (\ref{P_m_n_asymptotic}) with the same $I$ as $L$ increases. Moreover, the curves obtained from (\ref{P_m_n_asymptotic}) are always above the corresponding curves obtained from (\ref{P_m_n_theorem}), which agrees with Theorem \ref{Theorem_asymptotic}. This implies that it is more difficult to launch a successful DSA on a finitely-long blockchain than on an infinitely-long blockchain. Moreover, when $I\ge0.5$, the DSA can always be launched successfully on an infinitely-long blockchain. However, it is seen from Fig. \ref{fig_vs} that when $I\ge0.5$, $\bbP_L(L_0-L_a+1,L_0)<1$ for a finitely-long blockchain, which indicates that finitely-long blockchains are more resistant to $51\%$ attacks than infinitely-long blockchains.

\begin{figure}[H]\centering   
	\includegraphics[width=0.5\textwidth]{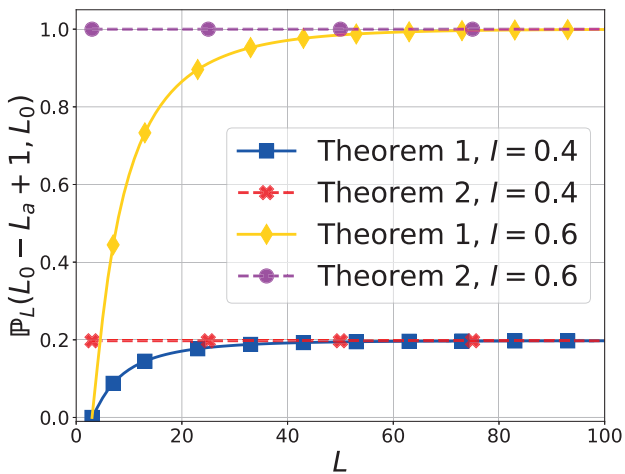}    
	\caption{The comparision between $\bbP_L(L_0-L_a+1,L_0)$  obtained from Theorem \ref{Theorem_DSA} and Theorem \ref{Theorem_asymptotic} as $L$ increases.}
	\label{fig_vs}
\end{figure}

\section{Conclusions}\label{Section_conclusion}

In this paper, we have theoretically studied the vulnerability of finitely-long blockchains in terms of securing data against double-spending attacks
We have developed a general closed-form expression for the probability that an attacker launches a successful double-spending attack on a finitely-long blockchain. This probability characterizes the vulnerability of  finitely-long blockchains in  securing data since if an attacker attempts to falsify the data which have already been stored in a blockchain, the attacker has to launch a successful double-spending attack on it. We have proven that finitely-long blockchains are less vulnerable to double-spending attacks than infinitely-long blockchains. Moreover, unlike infinitely-long blockchains, even though the normalized hash rate of an attacker is greater than $50\%$, the probability of success in launching a double-spending attack on any finitely-long blockchain is strictly less than one. This indicates that $51\%$ attacks cannot completely demolish the security of finitely-long blockchains. 


\appendices



%
%

\ifCLASSOPTIONcaptionsoff
  \newpage
\fi
\bibliographystyle{IEEEtran}
\bibliography{IEEEabrv,Blockchain}

\end{document}